\documentclass[pra, notitlepage, 10pt, aps]{revtex4-1}
\usepackage{amsmath}
\usepackage{amsfonts}
\usepackage{amssymb}
\usepackage{amsthm}
\usepackage{amsbsy}
\usepackage{cancel}
\usepackage[toc,page]{appendix}
\usepackage{xcolor}
\usepackage{graphicx}
\usepackage{ulem}
\usepackage{comment} 

\newcommand{\basis}{\mathcal{B}}
\newcommand{\states}{\mathcal{S}}
\newcommand{\effects}{\mathcal{E}}
\newcommand{\innerproduct}[2]{\left\langle #1 , #2 \right\rangle}
\newcommand{\adjoint}[1]{#1^{\dagger}}

\newtheorem{theorem}{Theorem}
\newtheorem{lemma}{Lemma}
\newtheorem{conjecture}{Conjecture}

\begin{document}

\title{Reversibility and the structure of the local state space}

\author{Sabri W. Al-Safi}
\email{Sabri.Alsafi@ntu.ac.uk}
\affiliation{School of Science \& Technology, Nottingham Trent University, Burton Street, Nottingham, NG1 4BU, UK}

\author{Jonathan Richens}
\email{Jonathan.Richens08@imperial.ac.uk}
\affiliation{Controlled Quantum Dynamics Theory Group, Department of Physics, Imperial College London, London SW7 2AZ}

\normalem

\begin{abstract}
The richness of quantum theory's reversible dynamics is one of its unique operational characteristics, with recent results suggesting deep links between the theory's reversible dynamics, its local state space and the degree of non-locality it permits. We explore the delicate interplay between these features, demonstrating that reversibility places strong constraints on both the local and global state space. Firstly, we show that all reversible dynamics are trivial (composed of local transformations and permutations of subsytems) in maximally non-local theories whose local state spaces satisfy a dichotomy criterion; this applies to a range of operational models that have previously been studied, such as $d$-dimensional ``hyperballs'' and almost all regular polytope systems. By separately deriving a similar result for odd-sided polygons, we show that classical systems are the only regular polytope state spaces whose maximally non-local composites allow for non-trivial reversible dynamics. Secondly, we show that non-trivial reversible dynamics do exist in maximally non-local theories whose state spaces are reducible into two or more smaller spaces. We conjecture that this is a necessary condition for the existence of such dynamics, but that reversible entanglement generation remains impossible even in this scenario.
\end{abstract}

\maketitle

\section{Introduction}

In the quest to understand why quantum mechanics accurately predicts natural phenomena, it is prudent to investigate the properties that distinguish it from classical mechanics and from other conceivable theories of nature. Exploring these properties leads to the development of algorithms for information-based tasks \cite{ekert91}, and provides insight into counter-intuitive quantum phenomena such as the prediction of non-local correlations \cite{bell64, pr94}, teleportation \cite{barnum08}, and the impossibility of cloning \cite{barnum06}. One property that seems particularly unique to quantum theory is reversibility: the fact that any two pure states of a system are connected by a continuous, reversible transformation. By considering the conservation of information, one might reasonably expect reversibility to hold for any physical theory. This is further supported by the fact that reversibility (or some variant of it) plays a pivotal role in information-theoretic reconstructions of quantum theory \cite{hardy01, hardy11, chiribella11, masanes11, masanes12, barnum14}.

By viewing quantum theory as one of a broad range of so-called \emph{general probabilistic theories} \cite{barrett05, barrett05.2, masanes06, dariano07, chiribella10, barnum12}, a growing number of recent results have hinted at a deep relationship between the local quantum state space and the property of reversibility. For example, hypothetical theories have been explored in which the local state space takes the form of a $d$-dimensional ball: for $d>3$, there can be no continuous, reversible interactions between two identical systems \cite{muller11}; for $d=3$ (the Bloch sphere), the only bipartite state space which allows for continuous, reversible interactions is given by the set of 2-qubit quantum states \cite{masanes12}. Therefore, out of all theories whose local state spaces are balls, only quantum theory is reversible. Demanding this local ball structure may itself be motivated on reasonable grounds (for example, by a modified form of Information Causality \cite{pawlowski09, masanes13}), thus these results are useful for generating minimal sets of principles which single out quantum theory.

In another popular hypothetical theory known as \emph{Boxworld}, every valid non-signaling outcome distribution over joint local measurements corresponds to an allowed state \cite{short&barrett10}. Boxworld is \emph{maximally non-local} - meaning that any composite state which is compatible with local states is allowed - thus giving rise to super-strong correlations such as the \emph{Popescu-Rohrlich box} \cite{pr94}. It is known that reversible dynamics in Boxworld are trivial: as long as no subsystem is classical, the only reversible transformations of a composite Boxworld system are composed of relabelings of measurement inputs and outputs, and permutations of subsystems \cite{colbeck10, alsafi14}. In particular, Boxworld is therefore not reversible, because a pure product state cannot be reversibly transformed into a Popescu-Rohrlich box. In fact, no correlations whatsoever can be reversibly generated between independent systems, demonstrating further that there is no ``Church of the larger Boxworld system'' in which a Boxworld measurement extends to a reversible transformation.

The mathematical structure of quantum theory differs from Boxworld in two significant ways. Firstly, Boxworld systems have only a finite number of pure states and hence their state space forms a convex polytope, whereas the quantum state space has an infinite number of pure states. Secondly, quantum theory is not maximally non-local, as there exist operators with negative global eigenvalues whose ``reduced states'' are nevertheless valid on their respective subsystems. We will demonstrate that Boxworld's trivial reversible dynamics in fact extends to all maximally non-local theories whose local state spaces satisfy a well-defined dichotomy criterion. Intuitively, dichotomic systems are those for which all maximally informative measurements have just two outcomes; this is similar to having the information capacity of a single bit, which has recently been studied as a fundamental postulate for local quantum systems \cite{masanes13}. In the case of polytopic state spaces, dichotomy is related to the geometric property of having diametrically opposed facets; in fact all  regular polytope state spaces except for $n$-simplexes (i.e. classical systems) and odd-sided polygons (introduced in \cite{janotta11}) are dichotomic. We provide a separate proof for the case of odd-sided polygons, demonstrating that classical theories are the only maximally non-local, reversible theories whose state spaces are regular polytopes. Many non-polytopic state spaces are also dichotomic, such as the $d$-dimensional balls mentioned above.

These results seem to suggest that reversible dynamics are always trivial in maximally non-local (and non-classical) theories. However, we will also demonstrate that classical correlations can in fact be reversibly generated if one of the local state spaces is reducible (or decomposable) into two or more smaller spaces. This is achieved via an analogue of the classical CNOT gate, where one controls on which of these smaller state spaces the state of the system is in. We leave as open questions whether non-local correlations can be reversibly generated when one or more subsystems are reducible, and whether the irreducibility of all subsystems is sufficient for all reversible dynamics to be trivial. Reducibility has previously arisen in the study of general probabilistic theories, although not in the context of reversibility \cite{barnum14}.

This article proceeds as follows: in Section \ref{set-up-section} we describe the formalism of general probabilistic theories; in Section \ref{transformations-section} we show how transformations are defined and give a useful necessary and sufficient condition for a reversible transformation to be trivial; in Section \ref{dichotomic-section} we show that all reversible dynamics are trivial in maximally non-local theories whose local systems are dichotomic; finally, in Section \ref{reducibility-section} we show that non-trivial transformations exist if one or more local systems are reducible, and conjecture that this is a necessary condition for reversible interactions in maximally non-local theories.

\section{Set-up \& Notation} \label{set-up-section}

In this Section we introduce the well-established framework of \emph{general probabilistic theories} which provides an operational formalism for modeling the observation of physical phenomena. This framework applies to almost any theory of nature involving systems whose states inform the outcome probabilities of future measurements. Conversely, it is straightforward to construct new probabilistic models which share many of the geometric features of quantum systems, and to investigate their operational properties with quantum theory. In standard quantum theory, a system is described by a complex Hilbert space $\mathcal{H}$, states correspond to density operators on $\mathcal{H}$ and effects correspond to positive operators $0 \leq E \leq \mathbb{I}$. Letting $V$ denote the real vector space of Hermitian operators on $\mathcal{H}$ equipped with the inner product $\innerproduct{A}{B} = Tr(AB)$, the sets of unnormalized states and effects are both identical to the self-dual cone in $V$ which is the set of positive operators. Moreover, the inner product between an effect and a state gives the probability of that effect occurring in a system which has been prepared in that state.

Under some basic assumptions, any system in a general probabilistic theory may be represented by means of a real, finite-dimensional inner product space $V$, in which the \emph{state space} $\mathcal{S}$ forms a compact, convex subset. The \emph{state cone} $\mathcal{S}_{+} \subseteq V$ is defined as the cone generated by $\mathcal{S}$, and the \emph{effect cone} $\mathcal{E}_+ \subseteq V$ is defined as the dual cone to $\mathcal{S}_+$. Conversely, $\states_+$ is the dual cone to $\effects_+$. We further assume that the cones $\states_+$ and $\effects_+$ are both pointed and generating, and that there exists a (unique) \emph{unit effect} $u \in \effects_+$ such that $\innerproduct{u}{e} = 1$ for all $s \in \states$.

A \emph{measurement} of the system consists of a set of effects $\{e_1 , \ldots , e_r\}$ which satisfy the normalization condition $\sum_i e_i = u$. For a system which has been prepared in state $s$, the probability of obtaining the outcome corresponding to $e_i$ is given by the inner product $\innerproduct{e_i}{s}$. Note that the normalization condition on the effects ensures that the outcome statistics of any measurement are likewise normalized.

The cone $\effects_+$ naturally induces a partial ordering on vectors in $V$: we say that $v \leq_{\effects_+} w$ if there exists some $e \in \effects_+$ such that $w = v + e$. If $e \leq_{\effects_+} f$ for effects $e, f$, we say that $e$ \emph{refines} $f$; for any measurement involving $f$, it is possible to replace $f$ by the two effects $e$ and $(f-e)$ and so obtain a new measurement which is at least as informative as the old one (if not more so). An effect $e$ is said to be \emph{proper} if $e \leq_{\effects_+} u$, i.e. $0 \leq \innerproduct{e}{s} \leq 1$ for all states $s \in \states$. We denote the set of proper effects by $\effects$, and note that it is a compact, convex subset of $V$.

The notion of a pure state in quantum theory has a natural analogy in the general probabilistic framework: a state $s$ is \emph{pure} if it is an extreme point of $\states$. Likewise, a proper effect is said to be \emph{extreme} if it is an extreme point of $\effects$, and \emph{ray-extreme} if it is extreme and generates an extreme ray of $\effects_+$. In quantum theory, pure states and ray-extreme effects are given by rank-one projectors, and extreme effects are given by projectors of any rank. Note that in 2-level quantum systems, the set of extreme effects (minus the zero and unit effects) coincides with the set of ray-extreme effects, whereas in higher-level systems there exist extreme effects which are not ray-extreme.

In this article we are interested in composite systems comprising $N \geq 2$ subsystems of the above type. A possible measurement on this composite system involves performing a local measurement individually on each subsystem. We assume that any state of the composite system may be characterized uniquely by the conditional probability distribution $P(a_1, \ldots , a_N | x_1 , \ldots , x_N)$ giving the probability of the outcomes $a_i$ occurring when the local measurements $x_i$ are performed separately on each subsystem. This assumption is often known as \emph{local tomography} \cite{masanes11,masanes12,barnum12,chiribella12}.

Local measurements on the subsystems may in principle represent physically separated events, hence we also assume that the choice of measurement on any single subsystem does not affect the marginal outcome statistics on the remaining subsystems, an assumption known as the \emph{non-signaling condition} \cite{barrett05}. This condition may be expressed mathematically as the requirement that for $i = 1 ,\ldots , N$, the following sum is independent of the value of $x_i$:
\begin{equation} \label{nonsignaling-condition}
\sum_{a_i} P(a_1 , \ldots , a_i , \ldots , a_n | x_1 , \ldots , x_i , \ldots , x_n ).
\end{equation}
The non-signaling condition implies that states of composite systems have well-defined reduced states which are obtained by ``tracing out'' one or more subsystems as in \eqref{nonsignaling-condition}. For consistency, we require that states of composite systems must have reduced states which correspond to genuine states on each local subsystem.

Suppose that for $i = 1, \ldots, N$, subsystem $i$ is represented by the vector space $V^{(i)}$, with state cone $\states^{(1)}_+$, effect cone $\effects^{(i)}_+$, and unit effect $u^{(i)}$. The \emph{max tensor product} of these subsystems is the set of all composite states which satisfy the non-signaling condition and local tomography, and whose reduced states are valid states of the local subsystems. The max tensor product can be neatly represented in the tensor product space $V = V^{(1)} \otimes \cdots \otimes V^{(N)}$ by defining the composite effect cone $\effects_+$ to be the cone generated by product effects $e = e^{(1)} \otimes \cdots \otimes e^{(N)}$, where $e^{(i)}$ (which we refer to as the $i$th component of $e$) is a member of $\effects^{(i)}$. The composite unit effect is given by $u = u^{(1)} \otimes \cdots \otimes u^{(N)}$. The composite state cone $\states_+$ is then defined as the dual cone to $\effects_+$, and the normalized states are those $s \in \states_+$ for which $\innerproduct{u}{s} = 1$. For a product effect $e$ and a subset $\Omega \subseteq \{1, \ldots , N\}$, it will be convenient to use the notation $e^{\Omega}$ to refer to the tensor product of those components of $e$ which belong to subsystems in $\Omega$, for example if $\Omega = \{1, 3, 4\}$ then $e^{\Omega} = e^{(1)} \otimes e^{(3)} \otimes e^{(4)}$. Note that $e^{\Omega}$ is itself a product effect in the reduced tensor product space $V^{\Omega} = V^{(1)} \otimes V^{(3)} \otimes V^{(4)}$.

A general probabilistic theory in which systems combine under the max tensor product is referred to as \emph{maximally non-local}. Note that any collection of local state spaces can be combined into a maximally non-local composite space, although the theory known as Boxworld is a canonical example of this \cite{short&barrett10}. In Boxworld, each subsystem $i = 1, \ldots , N$ is equipped with a finite set of fiducial measurements indexed by $x_i$, and each measurement choice gives rise to a finite set of outcomes indexed by $a_i$. Any non-signaling conditional probability distribution $P(a_1, \ldots , a_N | x_1 , \ldots x_N)$ then corresponds to an allowed state on the composite system. It is well known that Boxworld allows for much stronger correlations and information-processing capabilities between distant parties than are achievable in quantum theory \cite{vandam05,buhrman06, almeida10}, but that the set of reversible dynamics is extremely restricted \cite{colbeck10, alsafi14}.

\begin{figure}[t] 
\centering
\includegraphics[scale=0.3]{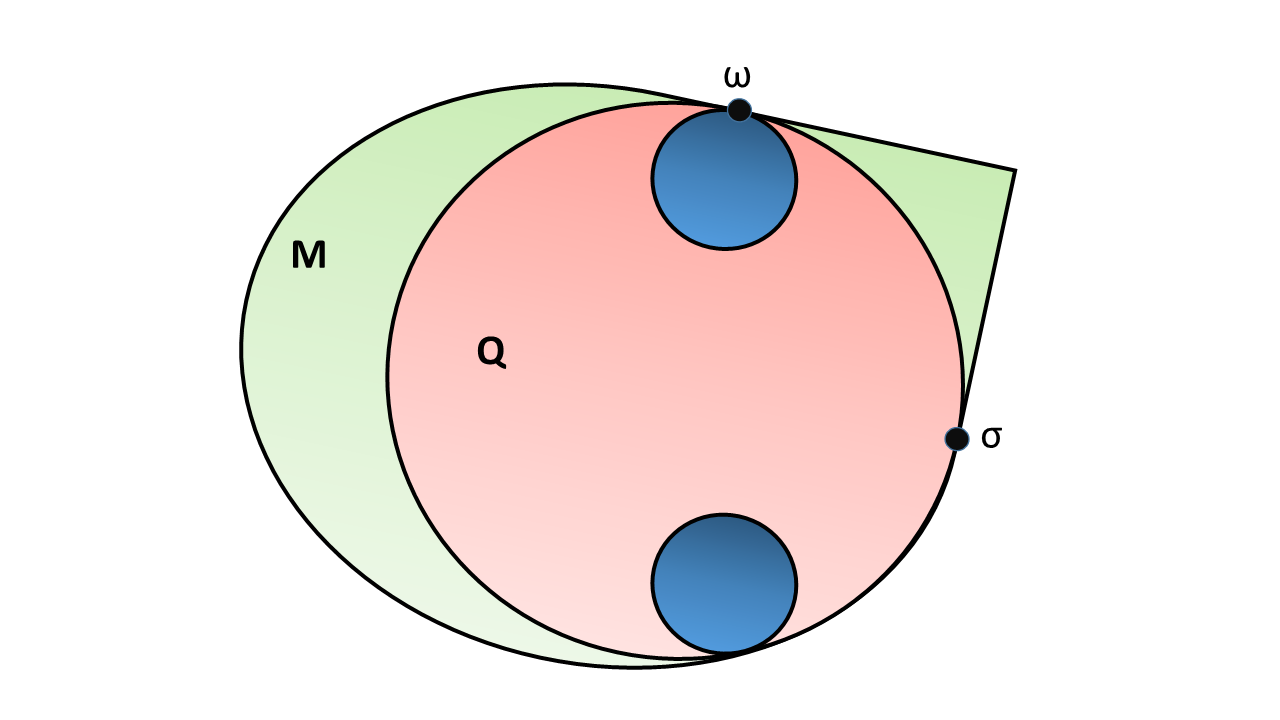}
\caption{Two possible composite state spaces: the local qubit state spaces (blue) can combine under either the standard quantum product \textbf{Q} (red), or the max tensor product \textbf{M} (green). The pure states $\omega$ and $\sigma$ represent a product state and a maximally entangled state respectively. The geometry of the composite state space influences the possible reversible dynamics; $\omega$ and $\sigma$ are linked by a reversible transformation of \textbf{Q}, but not by any reversible transformation of \textbf{M}. (Colour online.)}
\label{bipartite-state-space-figure}
\end{figure}

In quantum theory on the other hand, systems do not combine under the max tensor product. Rather, the composite state and effect cones are both given by the set of positive matrices over the tensor product of the Hilbert spaces representing each subsystem. The quantum cone strictly contains the max tensor product effect cone $\effects_+$, as not all of its extreme rays are tensor products of local projectors. The quantum cone is strictly smaller than the max tensor product state cone $\states_+$, as it does not contain ``entanglement witness'' states, i.e. operators which have negative eigenvalues, but which have positive inner product with any tensor product of local effects (Fig. \ref{bipartite-state-space-figure}). Exploring how fundamental physical concepts like reversibility break down in max tensor product theories like Boxworld provides insight into what principles, beyond local tomography and the non-signaling condition, constrain the set of quantum-achievable correlations.

We now introduce some further terminology of central importance in the discussion of reversible dynamics. A \emph{composite ray-extreme effect} is a tensor product of local ray-extreme effects, i.e. $e = e^{(1)} \otimes \cdots \otimes e^{(N)}$, where each $e^{(i)} \in \effects_+^{(i)}$ is ray-extreme.  We say that two composite ray-extreme effects are \emph{adjacent} if they differ on exactly one subsystem, and that they are \emph{adjacent on subsystem i} if it is subsystem $i$ on which they differ. For example, the effect $f = f^{(1)} \otimes e^{(2)} \otimes \cdots \otimes e^{(N)}$ is adjacent on subsystem 1 to the effect $e$ above, as long as $f^{(1)}$ is a ray-extreme effect distinct from $e^{(1)}$. A \emph{sub-unit effect} $E$ is a product effect whose $i$th component is $u^{(i)}$ for some $1\leq i \leq N$, and whose $j$th component for $j\neq i$ is some ray-extreme effect $e^{(j)}$, i.e. $E = e^{(1)} \otimes \cdots \otimes u^{(i)} \otimes \cdots \otimes e^{(N)}$. We say a sub-unit effect $E$ is an \emph{i-sub-unit} effect if $E^{(i)} = u^{(i)}$. Intuitively, a sub-unit effect corresponds to a ray-extreme effect of the reduced system after subsystem $i$ has been ``traced out''. The following result applies to any general probabilistic theory regardless of the structure of individual systems, and whether or not those systems combine under the max tensor product.

\begin{lemma} \label{leq-sub-unit-lemma}
Let $f$ and $g$ be distinct composite ray-extreme effects which both refine the same $i$-sub-unit effect $E$. Then $f$ and $g$ are adjacent at subsystem $i$.
\end{lemma}

\begin{proof}
If $f$ and $g$ are not adjacent at subsystem $i$, then without loss of generality there is some $j \neq i$ for which $f^{(j)} \neq E^{(j)}$. Since $f^{(j)}$ and $E^{(j)}$ are distinct ray-extreme effects on subsystem $j$, we have $f^{(j)} \nleq_{\effects_+} E^{(j)}$, hence there exists a local pure state $s^{(j)}$ for which $\innerproduct{f^{(j)}}{s^{(j)}} > \innerproduct{E^{(j)}}{s^{(j)}}$. For all remaining subsystems $k\neq j$, let $s^{(k)}$ be any pure state for which $\innerproduct{f^{(k)}}{s^{(k)}} \geq \innerproduct{E^{(k)}}{s^{(k)}}$. Then the pure product state $s = \bigotimes_{i=k}^{N}s^{(k)}$ satisfies
\begin{equation}
\innerproduct{f}{s} = \innerproduct{f^{(j)}}{s^{(j)}} > \innerproduct{E^{(j)}}{s^{(j)}} \geq \innerproduct{E}{s} .
\end{equation} 
This implies that $f \nleq_{\effects_+} E$, thus the result follows.
\end{proof}

\section{Transformations} \label{transformations-section}

In this Section we discuss the reversible dynamics of systems in general probabilistic theories. Given systems $V^{(1)}$ and $V^{(2)}$ with state spaces $\states^{(1)}$ and  $\states^{(2)}$, an allowed transformation $T$ from $V^{(1)}$ to $V^{(2)}$ is given by a mapping of $\states^{(1)}$ into $\states^{(2)}$. By considering probabilistic mixtures of states, it may be assumed that $T$ is \emph{convex-linear}, i.e. for any $s_1, s_2 \in \states$ and $0 \leq p \leq 1$, $T(ps_1 + (1-p) s_2) = pT(s_1) + (1-p) T(s_2)$. This assumption, along with the fact that $\states^{(1)}$ lies in the hyperplane of vectors which have unit inner product with the unit effect $u^{(1)}$, allows $T$ to be extended to a full linear map on $V^{(1)}$ \cite{barrett05}. $T$ is \emph{reversible} if this linear map has an inverse $T^{-1}$ which is also an allowed transformation; in this case we say that $V^{(1)}$ and $V^{(2)}$ are \emph{equivalent} systems. In this Section we will be concerned with transformations mapping a system $V$ to itself.

From an operational perspective, transformations are characterized by how they affect the outcome probabilities of later measurements. Given that $\innerproduct{e}{T(s)} = \innerproduct{\adjoint{T}(e)}{s}$, a transformation may equivalently be described via the action of the adjoint $\adjoint{T}$ on the set $\effects$. If $T$ is reversible, it is not hard to show by linearity that $\adjoint{T}$ maps ray-extreme effects to ray-extreme effects. Conversely, any linear map which permutes the set of ray-extreme effects and maps the unit effect to the unit effect is the adjoint of an allowed reversible transformation.

Whilst the above comments apply to any general probabilistic theory, in the remainder of this Section we are concerned with reversible transformations acting on a composite system $V$ which is the max tensor product of subsystems $V^{(1)}, \ldots , V^{(N)}$. Note that in this case, $\adjoint{T}$ acts as a permutation on the set of composite ray-extreme effects. Two classes of reversible transformations naturally arise in this setting. Firstly, $T$ is a \emph{local} transformation if there is a reversible transformation $T^{(i)} : V^{(i)} \rightarrow V^{(i)}$ such that
\begin{align}
\adjoint{T}( &e^{(1)} \otimes \cdots \otimes e^{(N)}) = \nonumber \\
&e^{(1)} \otimes \cdots \otimes \adjoint{ \left[ T^{(i)} \right] }(e^{(i)}) \otimes \cdots \otimes e^{(N)} .
\end{align}
Secondly, $T$ is a \emph{permutation} of subsystems $i$ and $j$ if there is a reversible linear map
\begin{equation}
P_{ij} : V^{(i)} \rightarrow V^{(j)}
\end{equation}
which maps $\effects^{(i)}$ bijectively onto $\effects^{(j)}$  (i.e. subsystems $i$ and $j$ are equivalent), such that
\begin{align}
\adjoint{T}( &e^{(1)} \otimes \cdots \otimes e^{(i)} \otimes \cdots \otimes e^{(j)} \otimes \cdots \otimes e^{(N)}) = \nonumber \\
&e^{(1)} \otimes \cdots \otimes P_{ij}^{-1} (e^{(j)}) \otimes \cdots \otimes P_{ij} (e^{(i)}) \otimes \cdots e^{(N)} .
\end{align}

A \emph{trivial} transformation is one that is a composition of local transformations and permutations of subsystems. Since local transformations are a special case of permutations of subsystems in which $i=j$, trivial transformations may be regarded simply as compositions of permutations of subsystems. Note that both these types of transformation map pure product states to pure product states, hence trivial transformations are incapable of generating even classical correlations between systems that have not previously interacted. 

In Boxworld, it has been shown that all reversible transformations are trivial, so long as none of the subsystems are classical \cite{colbeck10}. This result makes use of a combinatorial argument concerning how $\adjoint{T}$ maps pairs of composite ray-extreme effects. In particular, the fact that $\adjoint{T}$  is adjacency-preserving as a permutation of the finite set of composite ray-extreme effects is sufficient to deduce that $T$ is trivial. In the following Lemma we modify this argument in order to apply it to the more general scenario involving the max tensor product of arbitrary systems. This generalization is necessary firstly because the number of local ray-extreme effects may no longer be finite, and secondly because there may no longer be be a natural way of identifying effects between equivalent subsystems.

\begin{lemma} \label{adjacency-preserving-lemma}
Allowed, reversible, adjacency-preserving transformations of the composite system $V$ are trivial.
\end{lemma}

\begin{proof}
We construct a trivial reversible transformation $P$, and show that $\adjoint{P}\adjoint{T}$ is the identity transformation. Since trivial transformations have trivial inverses, it follows that $T$ itself is trivial.

Fix a composite ray-extreme effect $e = e^{(1)} \otimes \cdots \otimes e^{(N)}$, and let $f = \adjoint{T}(e)$. For $i = 1, \ldots , N$, let $n_i = \text{dim}(V^{(i)})$ and construct a basis $\basis^{(i)} = \left\{e^{(i)}_j\right\}_{j=1}^{n_i}$ of $V^{(i)}$ consisting of local ray-extreme effects, for which $e^{(i)}_1 = e^{(i)}$. Define
\begin{equation}
\hat{e}_{ij} = e^{(1)} \otimes \cdots \otimes e^{(i)}_j \otimes \cdots \otimes e^{(N)}
\end{equation}
and denote by $R_i$ the set of $n_i - 1$ basis vectors $\left\{\hat{e}_{ij}\right\}_{j=2}^{n_i}$ which are adjacent to $e$ on subsystem $i$. As $\adjoint{T}$ is adjacency-preserving, it maps $R_i$ to a set of vectors which are adjacent to $f$ on some subsystem which we will denote $\sigma(i)$. Moreover, the set $R = \bigcup_{i} R_i$ must be mapped to a set of $\sum_{i}(n_i - 1)$ linearly independent vectors, all of which are adjacent to $f$. This is only possible if $\sigma$ is a permutation of $\{1, \ldots , N\}$ and $n_{\sigma(i)} = n_i \; \forall i$.

Define $P_{1\sigma(1)} : V^{(1)} \rightarrow V^{(\sigma(1))}$ to be the linear extension of
\begin{equation}
e^{(1)}_j \rightarrow \left[ \adjoint{T}(\hat{e}_{1j})\right]^{\sigma(1)} \;\; \text{for } j = 1, \ldots , n_1.
\end{equation}
Note that $P_{1\sigma(1)}$ is reversible and maps $\effects^{(1)}$ bijectively onto $\effects^{(N)}$. Therefore $P_{1\sigma(1)}$ induces a permutation $P_1$ of subsystems 1 and $\sigma(1)$, such that $\adjoint{P_1}\adjoint{T}$ fixes the first component of each element of $R_1\cup\{e\}$. The above process may be repeated to obtain  $P_2$ such that $\adjoint{P_2}\adjoint{P_1}\adjoint{T}$ fixes the first component of each element of $R_1\cup\{e\}$ as well as the second component of each element of $R_2\cup\{e\}$.

After at most $N$ steps, we eventually construct a trivial transformation $P$ such that $\adjoint{P}\adjoint{T}$ is adjacency-preserving and fixes the $i$th component of each $\hat{e}_{ij}$. We argue that $\adjoint{P}\adjoint{T}$ is the identity, by considering its action on elements of the tensor-product basis $\basis$ arising from the bases $\basis^{(i)}$ for each subsystem $i$. Note that $\adjoint{P}\adjoint{T}$ fixes $e$, and therefore fixes every member of $\basis$ which is adjacent to $e$. Observe that any $g \in \basis$ is uniquely defined by the set of effects $g' \in \basis$ such that $d_H(g',e) = d_H(g,e)-1$ and $d_H(g',g) = 1$. By induction on $d_H(g,e)$, it is clear that $\adjoint{P}\adjoint{T}$ fixes every element of $\basis$, and hence is the identity transformation.
\end{proof}

\begin{lemma} \label{image-sub-unit-lemma}
Let $T$ be an allowed, reversible transformation of the composite system $V$. Then $T$ is trivial if and only if the image of every sub-unit effect refines a sub-unit effect.
\end{lemma}

\begin{proof}
For the ``if'' direction, consider two composite ray-extreme effects $e$ and $f$ which are adjacent on some subsystem $i$. Observe that there is a unique $i$-sub-unit effect $E$ which is refined by both $e$ and $f$. If $\adjoint{T}(E)$ refines a sub-unit effect $F$, then by linearity $\adjoint{T}(e)$ and $\adjoint{T}(f)$ also refine $F$. It follows from Lemma \ref{leq-sub-unit-lemma} that $\adjoint{T}(e)$ and $\adjoint{T}(f)$ are adjacent, hence by Lemma \ref{adjacency-preserving-lemma}, $T$ is trivial.

To prove the ``only if'' direction, observe that local reversible transformations and permutations of subsystems clearly map sub-unit effects to sub-unit effects, therefore any composition of them satisfies the desired condition.
\end{proof}

\section{Dichotomic systems} \label{dichotomic-section}

A consequence of Lemma \ref{leq-sub-unit-lemma} is that the local state space determines the ways in which a sub-unit effect can be written as a positive linear combination of composite ray-extreme effects. Transformations must respect these combinations, therefore Lemma \ref{image-sub-unit-lemma} suggests that the local state space also determines the form of reversible transformations. In this Section we consider \emph{dichotomic systems} in which for every local ray-extreme effect $e$, the effect $\bar{e} = u-e$ is also ray-extreme (note that for any theory,  $u-e$ is extreme whenever $e$ is - see e.g. Proposition 3.33 of \cite{pfister12}). This has a deep physical interpretation in terms of possible measurements. A \emph{fine-grained} measurement is one that consists solely of ray-extreme effects, and hence is maximally informative in the sense that none of its constituent effects $e$ can be replaced by two non-parallel effects $f$ and $g$ for which $f+g = e$ \cite{short10}. In dichotomic systems, all fine-grained measurements consist of just two outcomes; such systems may therefore be regarded as a fundamental unit of information in the theory.

It turns out that a surprisingly large range of operational models are represented by dichotomic state spaces. For example, it can be checked that many of the regular $d$-polytopes give rise to dichotomic state spaces. This follows from the result that an effect generates an extreme ray of $\effects_+$ as long as it has zero inner product with at least $d$ vertices of a $d$-polytope state space (see e.g. Theorem 2.16 of \cite{gale60}).  When $d=2$ the state space is a regular polygon, the non-local correlations of which have previously been studied in \cite{janotta11}; even-sided polygons are dichotomic, whereas odd-sided polygons are not. Of the five regular 3-polytopes, all except the 3-simplex (or tetrahedron) give rise to dichotomic state spaces; similarly for the six regular 4-polytopes, all except the 4-simplex do. For all $d>4$ there are exactly 3 regular $d$-polytopes: the $d$-simplex, the $d$-cube and the $d$-octoplex. The $d$-simplex represents a classical state space on $d$ outcomes, and is never dichotomic for $d>2$. The $d$-cube represents a Boxworld system with $d$ possible measurements and two possible outcomes, and is always dichotomic. The $d$-octoplex is the dual polytope to the $d$-cube and is also dichotomic.

Many non-polytopic systems are also dichotomic, for example qubit systems. In fact, any system for which the extreme effects are also ray-extreme is dichotomic, including all $d$-dimensional ball systems, i.e. whose state space is an embedding into $\mathbb{R}^{d+1}$ of the set of vectors $s \in \mathbb{R}^d$ for which $||s|| \leq 1$ \cite{muller11, pawlowski12}. However, quantum systems whose Hilbert space dimension is greater than 2 are not dichotomic, since subtracting a rank-one projector from the identity produces an effect that is extreme but not ray-extreme.

In the remainder of the Section we prove that all reversible transformations on a max tensor product of dichotomic systems satisfy the conditions of Lemma \ref{image-sub-unit-lemma}, therefore are trivial. For completeness, a separate proof of this result for the case of identical odd-sided polygon systems is included in the Appendix.

\begin{lemma}\label{tensor-product-lemma}
Suppose that $\{ x_i \}_{i=1}^{r}, \{ w_j \}_{j=1}^{s} \subset U$ and $\{ y_i \}_{i=1}^{r}, \{ z_j \}_{j=1}^{s} \subset W$ are sets of vectors satisfying:
\begin{equation}
\sum_{i=1}^{r} x_i \otimes y_i = \sum_{j=1}^{s} w_j \otimes z_j ,
\end{equation}
and that the set $\{ x_i \}_{i=1}^{r}$ is linearly independent. Then,
\begin{equation}
y_i \in \text{span}\left( \{ z_1, \ldots ,  z_{s} \} \right) .
\end{equation}
for $i = 1, \ldots , r$.
\end{lemma}

\begin{proof}
See Lemma 1 of \cite{marcus59}.
\end{proof}

\begin{lemma} \label{ray-span-lemma}
Suppose that $e, f$ and $g$ are ray-extreme effects in a (local or composite) system, with $f \neq g$. If $e \in \text{span}\left( \{ f, g \}\right)$, then either $e = f$ or $e = g$.
\end{lemma}
\begin{proof}
Let $e = \alpha f + \beta g$. As $e$ is ray-extreme, we cannot have both $\alpha > 0$ and $\beta > 0$. We also cannot have both $\alpha < 0$ and $\beta < 0$, otherwise any state $s$ for which $\innerproduct{f}{s} > 0$ gives $\innerproduct{e}{s} < 0$. Without loss of generality, assume $\alpha \leq 0$ and $\beta \geq 0$. Then $e + (-\alpha) f = \beta g$, implying that either $\alpha = 0$ (in which case $e=g$) or $\beta = 0$ (in which case $e=f$). 
\end{proof}

\begin{theorem} \label{dichotomic-theorem}
All reversible transformations on a max tensor product of non-classical, dichotomic systems are trivial.
\end{theorem}

\begin{proof}
Let $E$ be a sub-unit effect in the composite system, and note that $E$ can be written in at least two distinct ways as a sum of two composite ray-extreme effects. As $\adjoint{T}$ is linear and permutes the set of composite ray-extreme effects, $\adjoint{T}(E)$ can also be written in at least two distinct ways as a sum of two composite ray-extreme effects:
\begin{align} \label{decomp-eqn}
\adjoint{T}(E) &= e^{(1)} \otimes \cdots \otimes e^{(N)} + f^{(1)} \otimes \cdots \otimes f^{(N)} \nonumber \\
&= g^{(1)} \otimes \cdots \otimes g^{(N)} + h^{(1)} \otimes \cdots \otimes h^{(N)} .
\end{align}
We claim that there exists a sub-unit effect $F$ for which $\adjoint{T}(E) \leq_{\effects_+} F$. It then follows from Lemma \ref{image-sub-unit-lemma} that $T$ is trivial. 

Firstly, suppose that $e^{(i)}=f^{(i)}$ for $i = 1, \ldots , N$. Then there exists a pure product state whose inner product with $\adjoint{T}(E)$ is 2, contradicting the fact that $\adjoint{T}$ maps proper effects to proper effects. Secondly, suppose that $e^{(i)} \neq f^{(i)}$ and  $e^{(j)} \neq f^{(j)}$, where $1 \leq i < j \leq N$. Let $\Omega_1 = \{1 , \ldots , i\}$ and $ \Omega_2 = \{ i+1 , \ldots , N\}$, and define 
\begin{equation}
\setlength\arraycolsep{20pt}
\begin{array}{cccc}
x_1 = e^{\Omega_1}, & y_1 = e^{\Omega_2}, & x_2 = f^{\Omega_1}, & y_2 = f^{\Omega_2} \\
w_1 = g^{\Omega_1}, & z_1 = g^{\Omega_2}, & w_2 = h^{\Omega_1}, & z_2 = z^{\Omega_2}, \\
\end{array}
\end{equation}
all of which are ray-extreme effects in one of the (possibly composite) max tensor product systems $U = V^{\Omega_1}$ or $W = V^{\Omega_2}$. Equation \eqref{decomp-eqn} then reduces to
\begin{equation}
x_1 \otimes y_1 + x_2 \otimes y_2 = w_1 \otimes z_1 + w_2 \otimes z_2.
\end{equation}
Note that by construction $x_1 \neq x_2$ and $y_1 \neq y_2$. As $x_1$ and $x_2$ are linearly independent, it follows from Lemma \ref{tensor-product-lemma} that $y_1 \in \text{span} \left( \{ z_1, z_2 \} \right)$, hence by Lemma \ref{ray-span-lemma} we may assume without loss of generality that $y_1 = z_1$. Similarly, $y_2 \in \text{span} \left( \{ z_1, z_2 \}  \right)$, and since $y_1 \neq y_2$, we have that $y_2 = z_2$. Thus Equation \eqref{decomp-eqn} reduces to
\begin{equation} \label{reduced-decomp-eqn}
x_1 \otimes y_1 + x_2 \otimes y_2 = w_1 \otimes y_1 + w_2 \otimes y_2.
\end{equation}
Let $b \in W$ be any vector such that $\innerproduct{y_1}{b} = 1$ and $\innerproduct{y_2}{b} = 0$. Then for arbitrary $a \in U$, by taking the inner product of both sides of Equation \eqref{reduced-decomp-eqn} with $a \otimes b$, we find that $\innerproduct{x_1}{a} = \innerproduct{w_1}{a}$. Since $a$ is arbitrary, it must be that $x_1 = w_1$, and by similar reasoning $x_2 = w_2$. This contradicts the fact that the two ways of writing $\adjoint{T}(E)$ in Equation \eqref{decomp-eqn} are distinct.

The only remaining possibility is that $e^{(i)} \neq f^{(i)}$ for exactly one value of $i$. In this case $\adjoint{T}(E)$ decomposes as
\begin{align}
\adjoint{T}(E) &= e^{(1)} \otimes \cdots \otimes \left[ e^{(i)} + f^{(i)} \right] \otimes \cdots \otimes e^{(N)} \nonumber\\
& \leq_+ e^{(1)} \otimes \cdots \otimes u^{(i)} \otimes \cdots \otimes e^{(N)},
\end{align}
where the inequality follows since, if $\adjoint{T}(E)$ is a proper effect, then $e^{(i)} + f^{(i)} \in \effects^{(i)}$. 

Having shown that the image of a sub-unit effect refines a sub-unit effect, it follows from Lemma \ref{image-sub-unit-lemma} that $T$ is trivial.
\end{proof}

\section{Reducibility} \label{reducibility-section}

In this Section we explore how reversibility relates to the mathematical property of reducibility. A closed, pointed cone $K \subseteq V$ with a set $E$ of extreme rays is said to be \emph{reducible} if there exists a decomposition $V = V_1 \oplus V_2$ such that each extreme ray of $K$ lies either in $V_1$ or $V_2$, i.e. $E = E_1 \cup E_2$ where $E_1 = E \cap V_1$ and $E_2 = E \cap V_2$. If $K$ is a reducible cone, we write $K = K_1 \oplus K_2$, where $K_i$ is the cone in $V_i$ generated by $E_i$ for $i=1,2$. We say that a system $V$ with state space $\states$ is reducible if the corresponding state cone $\states_+$ is a reducible cone. Since the dual cone of a closed, generating, reducible cone is also reducible (for example, see Theorem 2.1 of \cite{barker75}), we may equivalently say that a system is reducible if its effect cone $\effects_+$ is reducible.

\emph{Examples of irreducible systems.} Polygon state spaces whose number of pure states $n$ exceeds 3 are irreducible. To see this for $n>4$ note that either $E_1$ or $E_2$ must contain at least 3 extreme rays of $\states_+$, and any 3 such vectors span $V$. A similar argument also shows that quantum systems are irreducible in any dimension. Dichotomic systems (including the $n=4$ polygon) are also irreducible. To see this, let $e\in E_1$ and $f \in E_2$; no matter how one assigns $\overline{e} = u - e$ and $\overline{f} = u - f$ to $E_1$ and $E_2$, a contradiction may easily be derived from the equality $e + \overline{e} = f + \overline{f}$. A similar argument also shows that all non-classical Boxworld systems are irreducible.

\emph{Examples of reducible systems.} Classical systems are represented by state cones whose extreme rays form a basis of $V$, and so are clearly reducible. A simple, non-classical, reducible system is the \emph{squashed g-trit} \cite{alsafi15}, a modification of a Boxworld system with 2 measurement choices and 3 outcomes, with the added constraint on the state space that the first outcome of both measurements must have the same probability. Such a system has extreme rays $E = \{X, Y_{0}, Y_{1}, Z_{0}, Z_{1}\}$ such that $\{u, X, Y_{0}, Z_{0}\}$ forms a basis of $V$ and $u = X + Y_{0} + Y_{1} = X + Z_{0} + Z_{1}$. Splitting the extreme rays into $E_1 = \{X\}$, $E_2 = E \setminus E_1$ then gives a valid decomposition of the effect cone. The corresponding set of normalized states then forms a square pyramid (Fig. \ref{squashed-g-trit-figure}).

\begin{figure}[t] 
\centering
\includegraphics[scale=0.4]{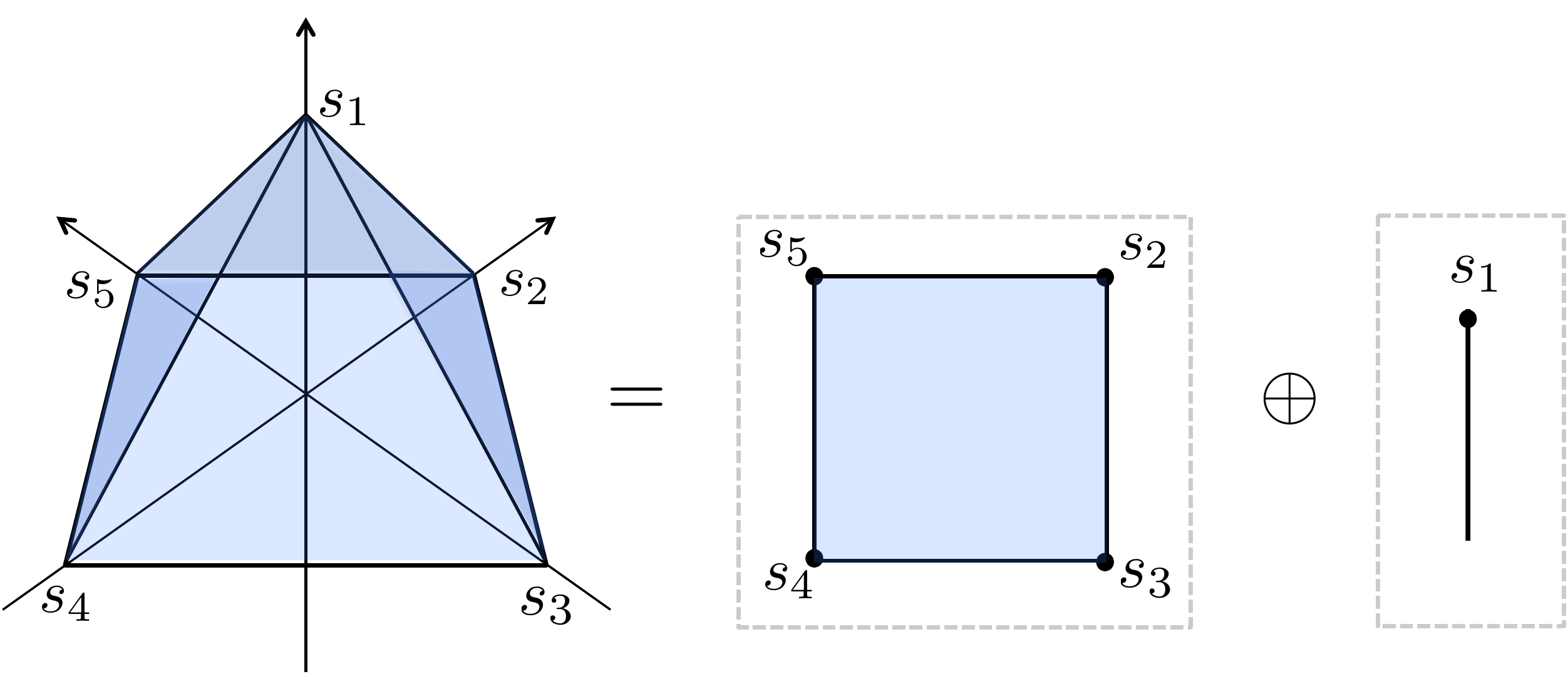}
\caption{State space of the squashed g-trit: State $s_1$ lies on the vertical axis, whilst the remaining pure states lie in the horizontal plane, forming a square pyramid. The pyramid is then embedded into the hyperplane in $\mathbb{R}^4$ which has inner product 1 with the unit effect $u$.}
\label{squashed-g-trit-figure}
\end{figure}

Reducible systems can be seen as carrying classical information, even if they are not in fact classical systems. Intuitively speaking, this classical information refers to whether the state of the system lies in $K_1$ or in $K_2$. By conditioning transformations of another subsystem on this pseudo-classical structure, non-trivial reversible transformations can be constructed in the max tensor product of composite systems. This type of transformation can be viewed as a generalization of the classical CNOT operation to more general systems.

\begin{theorem} \label{reducibility-theorem}
Let $V$ be the max tensor product of systems $V^{(1)}, \ldots , V^{(n)}$. If there exists some $i$ such that $V^{(i)}$ is a reducible system, and some $j \neq i$ such that the system $V^{(j)}$ has at least one local reversible transformation, then there exist non-trivial reversible transformations on $V$.
\end{theorem}

\begin{proof}
Suppose without loss of generality that $V^{(1)}$ is reducible, i.e. $\effects^{(1)}_+ = K_1^{(1)} \oplus K_2^{(1)}$ with $K_i^{(1)} \subseteq V_i^{(1)}$, and that $T^{(2)} \neq \mathbb{I}$ is a local reversible transformation on subsystem $V^{(2)}$ with inverse $S^{(2)}$. For $i=1,2$, fix a basis $\basis^{(1)}_i$ of $V_i^{(1)}$ consisting of ray-extreme effects. For the remaining subsystems $V^{(j)}$, $j=2, \ldots , N$, fix a basis $\basis^{(j)}$ consisting similarly of ray-extreme effects. Define $\adjoint{T}$ to be the linear extension of the following action on tensor products of basis vectors:
\begin{align}
\adjoint{T} ( &e^{(1)} \otimes e^{(2)} \otimes \cdots \otimes e^{(N)} ) = \nonumber \\
&\left\{
\begin{array}{ll}
e^{(1)} \otimes \adjoint{ \left[ T^{(2)} \right] } (e^{(2)}) \otimes \cdots \otimes e^{(N)} & \text{if } e^{(1)} \in \basis^{(1)}_1 \\
e^{(1)} \otimes e^{(2)} \otimes \cdots \otimes e^{(N)} & \text{if } e^{(1)} \in \basis^{(1)}_2
\end{array}
\right.
\end{align}
It remains to show that $T$ is reversible, allowed, and non-trivial. To see that it is reversible, observe that we may construct $\left[\adjoint{T} \right]^{-1}$ by mapping $e^{(2)}$ to $\left[\adjoint{S} \right]^{(2)}(e^{(2)})$, conditional on $e^{(1)} \in \basis^{(1)}_1$.

To show that $T$ is allowed, it is sufficient to show that $\adjoint{T}$ permutes the set of composite ray-extreme effects and leaves $u$ invariant. Let $f$ be a composite ray-extreme effect, and note that $f^{(1)}$ lies in the span of either $\basis^{(1)}_1$ or $\basis^{(1)}_2$, but not both. The fact that $\adjoint{T}(f)$ is also a composite ray-extreme effect follows from expanding $f$ in the tensor-product basis and applying the linearity of $\adjoint{T}$ and $\adjoint{ \left[ T^{(2)} \right] }$. 

Now let $u^{(1)} = \sum_i e^{(1)}_i + \sum_j \tilde{e}^{(1)}_j$, where $e^{(1)}_i \in \basis^{(1)}_1$ and $\tilde{e}^{(1)}_j \in \basis^{(1)}_2$. Then
\begin{align}
\adjoint{T} (u \otimes \cdots \otimes u ) = &\sum_i \adjoint{T} (e^{(1)}_i \otimes u^{(2)} \otimes \cdots \otimes u^{(N)}) \nonumber \\
&+ \sum_j \adjoint{T} (\tilde{e}^{(1)}_j \otimes u^{(2)} \otimes \cdots \otimes u^{(N)}) \label{t2-unit} \\
= &u^{(1)} \otimes \cdots \otimes u^{(N)},
\end{align}
where we have used the fact that $\adjoint{\left[T^{(2)} \right]}(u^{(2)}) = u^{(2)}$. Therefore $\adjoint{T}$ preserves the unit effect $u$.

Finally, we argue that $T$ is non-trivial by showing that $\adjoint{T}$ is not adjacency-preserving. Let $e^{(1)} \in \basis^{(1)}_1$, $f^{(1)} \in \basis^{(1)}_2$, and let $e^{(2)}$ be a ray-extreme effect on subsystem 2 which is not left invariant by $\adjoint{\left[T^{(2)} \right]}$. For $i \geq 3$ let $e^{(i)}$ be any ray-extreme effect. Then the adjacent effects  $f^{(1)}\otimes e^{(2)} \otimes \cdots \otimes e^{(N)}$ and $e^{(1)}\otimes e^{(2)} \otimes \cdots \otimes e^{(N)}$ are mapped to non-adjacent effects.
\end{proof}

Following the proof of Theorem \ref{reducibility-theorem}, one can construct non-trivial reversible transformations of the max tensor product of two squashed g-trit systems. For example, there is a transformation which acts in the following way on composite ray-extreme effects:
\begin{align}
A \otimes X &\mapsto A \otimes X , \nonumber \\
A \otimes B_i &\mapsto 
\left\{
	\begin{array}{ll}
		A \otimes B_{i\oplus 1} & \mbox{if } A = X \\
		A \otimes B_{i} & \mbox{otherwise} 
	\end{array}
\right.
\end{align}
where $A$ is any ray-extreme effect on subsystem 1 and $B$ may be replaced by either $Y$ or $Z$.

Theorem \ref{reducibility-theorem} demonstrates that the presence of a reducible subsystem is a sufficient condition for a maximally non-local system to admit non-trivial reversible transformations. We conjecture, based on Theorem \ref{dichotomic-theorem} and the results of \cite{alsafi14}, that this is also a necessary condition, i.e. that the existence of non-trivial reversible transformations is equivalent to the reducibility of one or more subsystems.

\begin{conjecture} \label{irreducibility-conjecture}
All reversible transformations in a max tensor product of irreducible systems are trivial.
\end{conjecture}

It may be observed that ``conditional'' transformations of the kind constructed in the proof of Theorem \ref{reducibility-theorem} can generate only classical correlations, rather than entanglement. This is because conditional transformations map pure product states to pure product states, and hence permute the set of classically correlated states that lie in the convex hull of the pure product states (analogous to separable states in quantum theory). However, it is conceivable that maximally non-local composite systems do allow for reversible entanglement generation, even if Conjecture \ref{irreducibility-conjecture} holds. This seems unlikely, given the highly restricted nature of reversible dynamics studied thus far, suggesting the following conjecture.

\begin{conjecture}
There is no reversible entanglement generation in maximally non-local systems.
\end{conjecture}

\section{Discussion} \label{discussion-section}

We have investigated reversible dynamics in maximally non-local general probabilistic theories, demonstrating in Section \ref{dichotomic-section} that reversible dynamics are trivial in the case where all subsystems are non-classical and dichotomic. This represents a major step in the study of reversibility: we have not needed the assumption made in \cite{masanes11, colbeck10} that all subsystems are identical, and unlike \cite{colbeck10, alsafi14} our results apply to a range of theories, including many whose local state spaces have an infinite number of pure states. Note that we actually only require the weaker condition that on each subsystem the unit effect can be decomposed in two distinct ways as a sum of two ray-extreme effects. A crucial step in our proof was Lemma \ref{adjacency-preserving-lemma}, which extends the Hamming-distance based argument of \cite{colbeck10} to account for arbitrary local state spaces, and provides a very useful tool in future work on reversibility. A natural extension of this result would be to relax the dichotomy criterion, for example replacing it with the condition that every ray-extreme effect is contained in a fine-grained measurement (not necessarily of size 2).

This result also has implications for information-theoretic reconstructions of quantum theory. As discussed at the beginning of Section \ref{dichotomic-section}, the existence of dichotomic systems in one's theory can be regarded as a physically motivated postulate. Our result shows that in order for reversible interactions to be possible between dichotomic systems, they cannot combine under the max tensor product; however, very few natural theories are known whose global state spaces are restricted in this way. This line of argumentation is similar to that of \cite{masanes13}, in which the existence of a fundamental information unit is central to reconstructing quantum theory. It can also be seen as a stronger version of the result that the max tensor product is unattainable if one demands that every logical bit can be reversibly mapped to any other logical bit \cite{muller12}. In the case of systems whose state spaces are regular polytopes, we have shown that only classical systems have non-trivial reversible interactions in the max tensor product. This provides a strong operational reason for why polytopic systems do not describe nature, which may be compared with the result that classical systems are the only polytopic systems in which deterministic measurements do not disturb the state \cite{pfister12}.

In Section \ref{reducibility-section} we demonstrated that it is possible to reversibly generate correlations in maximally non-local theories whose subsystems are not classical, as long as at least one of them is reducible. Note that it is possible to maximally violate Bell inequalities in such theories \cite{alsafi15}, therefore our results do not imply that such systems are reversible. This result hints at an underlying relationship between local reducibility and global reversibility, and we have conjectured that non-trivial transformations can \emph{only} occur if at least one subsystem is reducible. One promising avenue towards a proof of this would be to exploit the intriguing fact that the reducibility of a cone $K$ is equivalent to the ray-extremality of all reversible matrices in $\Gamma(K)$, the cone of matrices which map $K$ into itself \cite{loewy75}.

We have further conjectured that reversible entanglement generation is impossible in maximally non-local theories, even if one or more of the subsystems are reducible. In technical terms, this is equivalent to the statement that reversible transformations of max tensor product systems permute the set of pure product states. In \cite{colbeck10} a very simple proof of this was given in the case of arbitrary Boxworld systems, by characterizing pure product states according to their inner products with ray-extreme effects and using the fact that the adjoint of a reversible transformation permutes the latter. It seems likely that this method of proof is extendable to more general systems, and possibly even to all maximally non-local theories; a first step may be to consider local systems for which each pure state $s$ has a unique effect $e$ for which $\innerproduct{e}{s} = 1$, such as $d$-dimensional balls. Combined with the previous conjecture, this would constitute a substantial characterization of reversible dynamics in maximally non-local theories.

As we have discussed already, a notable difference between quantum theory and the theories studied in this article is that quantum systems do not combine under the max tensor product. Therefore it would be interesting to explore whether these results can be extended to theories which are not maximally non-local. The biggest obstacle to this is that the set of composite ray-extreme effects no longer generate all the extreme rays of the effect cone, so that reversible transformations may well map these effects to other types of effects (as happens in quantum theory). Intuitively it seems that a great deal of symmetry is necessary for this to be possible; it may be that some minor assumptions on the local state space, for example that the number of pure states is finite, will suffice to extend our results to this case. This line of research will go some way towards settling the open question of whether quantum theory is essentially unique in its continuous reversibility, which would be a major result in our operational understanding of the universe. \vspace{2pt}

\textbf{Acknowledgements.}
SWA thanks Anthony Short, Jonathan Barrett and Marcus M\"uller for noting the importance of reducibility in the squashed g-trit example. JR thanks Lluis Masanes for helpful discussions. JR is supported by EPSRC.

\section*{Appendix}

In this appendix we extend the results of Section \ref{dichotomic-section} to prove that there are no non-trivial reversible  transformations in the max tensor product of identical, odd-sided polygon systems. We generalize a technique developed in \cite{colbeck10}, employing a specific representation of states and effects for which all reversible transformations correspond to orthogonal maps. The proof of Theorem \ref{dichotomic-theorem} relies heavily on the dichotomy of the state space, so it is tempting to postulate that for state spaces that are not dichotomic (e.g. for which there are extreme effects that are not ray-extreme) we may be able to reversibly generate entanglement. The odd-sided regular polygons describe such a class of state spaces. These state spaces are strongly self-dual (in fact they satisfy the stronger condition of \emph{bit symmetry} \cite{muller12}) and generate bipartite correlations that obey Tsirelson's bound \cite{janotta11}, making them particularly attractive candidates for building a reversible non-local toy model. We use a modification of the representation for local effects introduced in \cite{janotta11}, in which the ray-extreme effects of odd-sided polygon systems are of the form,
\begin{equation} \label{rep1}
e_i = \frac{1}{1 + r^2_n}
\begin{pmatrix}
1 \\
r_n s_n(i) \\
r_n c_n(i) \\
\end{pmatrix}, \qquad (i = 1,\ldots,n)
\end{equation}
where $s_n(i) = \sin \left(\frac{2 \pi i }{n} \right)$, $c_n(i) = \cos \left(\frac{2 \pi i }{n} \right)$, and $r_n = \sqrt{\sec (\pi / n)}$.
\begin{lemma}\label{orthogonalitylemma}
For all polygon state spaces there exists a parametrization of the states and effects such that all reversible transformations are orthogonal
\end{lemma}

\begin{proof}
for a state space consisting of $N$ polygons with the $k^\text{th}$ polygon having $n(k)$ vertices, define
\begin{equation}
\Lambda = \bigotimes_{k=1}^{N} \left( 1+r_{n(k)}^2\right)\begin{pmatrix}
1 & 0 & 0 \\
0 & \sqrt{2}/r_{n(k)} & 0 \\
0 & 0 & \sqrt{2}/r_{n(k)} \\
\end{pmatrix} .
\end{equation}
Using $\Lambda$ we can define our reparametrized effect vectors as $\tilde{e}_i = \Lambda (e_i )$ and state vectors as $\tilde{s}_i = \Lambda^{-1} (s_i)$. Note that 
\begin{equation}
\langle \tilde{e}_i , \tilde{s}_j \rangle = \langle (\Lambda^{-1})^\dagger\Lambda e_i , s_j \rangle = \langle e_i , s_j \rangle ,
\end{equation}
as $\Lambda$ is diagonal. Therefore this parametrization is operationally equivalent to the parametrization \eqref{rep1}. The local ray-extreme effects are now of the form 
\begin{equation}
\tilde{e}_i = 
\begin{pmatrix}
1 \\
\sqrt{2} \, s_n(i) \\
\sqrt{2} \, c_n(i) \\
\end{pmatrix}, \qquad (i = 1,\ldots,n) .
\end{equation}
In this representation 
\begin{equation}
\sum\limits_i \tilde{e}^{(k)}_i\tilde{e}^{(k)\dagger}_i = n(k) \mathbb{I}_{3\times 3} ,
\end{equation}
which follows from the identities 
\begin{align}
&\sum_i s_n(i) = \sum_i c_n(i) = \sum_i s_n(i) c_n(i) = 0 \nonumber \\ 
&\sum_i s_n^2(i) = \sum_i c_n^2(i) = n/2 .
\end{align}
Let $\{e_i\}$ be an enumeration of the composite ray-extreme effects, with $e_i = \tilde{e}^{(1)}_i\otimes\cdots\otimes \tilde{e}^{(N)}_i$. Then, 
\begin{align}\label{identity-equation}
\sum\limits_i e_i e_i^\dagger &= \left(\sum\limits_i \tilde{e}^{(1)}_i\tilde{e}^{(1)\dagger}_i\right)\otimes\cdots\otimes \left(\sum\limits_j \tilde{e}^{(N)}_j\tilde{e}^{(N)\dagger}_j\right) \nonumber \\
&= \bigotimes_{k=1}^{N} n(k)\mathbb{I}_{3\times 3}= \left(\prod\limits_{k=1}^N n(k) \right) \mathbb{I} .
\end{align}
Reversible transformations permute the set of ray-extreme effects, therefore
\begin{equation} \label{permutation-equation}
T^\dagger \left( \sum\limits_i e_i e_i^\dagger \right) T = \sum\limits_i e_i e_i^\dagger .
\end{equation}
From \eqref{identity-equation} and \eqref{permutation-equation} it follows that $T^\dagger T = \mathbb{I}$ .
\end{proof}
Because $T^\dagger$ is orthogonal it preserves the inner product between all ray-extreme effects. We will exploit this fact, along with the inner products between local ray-extreme effects, to show that $T^\dagger$ is adjacency-preserving on composite ray-extreme effects. For simplicity we focus on the simplest case of composite polygon systems consisting of identical subsystems. The inner product between two local ray-extreme effects is given by 
\begin{equation}
\langle \tilde{e}_i, \tilde{e}_j \rangle = 1+ 2 \cos \left(\frac{2\pi (i-j)}{n} \right)\, , \quad (i-j)=0,\dots,n-1 .
\end{equation}
The inner products of interest for our proof are:
\begin{align}
\langle \tilde{e}_i , \tilde{e}_i \rangle &= 3\, ,  \quad &&\forall \,  i, \nonumber \\ 
\langle \tilde{e}_i , \tilde{e}_{j} \rangle &= C_\text{max}\, , \quad && j = i\pm 1, \nonumber \\
\langle \tilde{e}_i , \tilde{e}_{j} \rangle &= C_\text{min}\, , \quad &&j= i + \frac{n\pm 1}{2}, \label{opposite}
\end{align}
where indexing is understood to be modulo $n$. Note that $C_\text{max}$ is the largest positive inner product between two local ray-extreme effects that are non-identical, and $C_\text{min}$ is the largest negative inner product between two local ray-extreme effects. If $\langle \tilde{e}_i , \tilde{e}_j \rangle = C_\text{max}$, i.e. $j=i\pm 1$, we say $\tilde{e}_i$ is \emph{neighboring} $\tilde{e}_j$, denoted $\tilde{e}_i \wedge \tilde{e}_j$. If $\langle \tilde{e}_i , \tilde{e}_j \rangle = C_\text{min}$, i.e. $j=i+\frac{n\pm 1}{2}$, we say $\tilde{e}_i$ is \emph{opposite} to $\tilde{e}_j$, denoted $\tilde{e}_i\vee \tilde{e}_j$.
\begin{theorem}
All reversible transformations on a max tensor product of identical, odd-sided, non-classical polygon systems are trivial.
\end{theorem}
\begin{proof}
For $n>3$ observe that $C_\text{max} \neq C_\text{min}$ and $|C_\text{max}|,|C_\text{min}|<3$. For neighboring local ray-extreme effects, i.e. $\tilde{e}_1$, $\tilde{e}_2$, there is a unique ray-extreme effect $\tilde{e}_s$ such that $\tilde{e}_s \vee \tilde{e}_1$ and $\tilde{e}_s \vee \tilde{e}_2$, given by $s = 3/2 + n/2$ (simultaneously satisfying the condition given in \eqref{opposite} for $i=1,2$). For arbitrary $1 \leq k \leq N$, consider the adjacent composite ray-extreme effects
\begin{align}
e_1 &= \tilde{e}^{(1)} \otimes\cdots\otimes  \tilde{e}^{(k)}_1\otimes \cdots\otimes \tilde{e}^{(N)} \nonumber \\ 
e_2 &= \tilde{e}^{(1)} \otimes\cdots\otimes  \tilde{e}^{(k)}_2\otimes \cdots\otimes \tilde{e}^{(N)} \nonumber \\ 
e_3 &= \tilde{e}^{(1)} \otimes\cdots\otimes  \tilde{e}^{(k)}_s\otimes \cdots\otimes \tilde{e}^{(N)} .
\end{align}
Note that $\langle e_1, e_2 \rangle = 3^{N-1}C_\text{max}$, the largest possible positive inner product between non-identical composite ray-extreme effects, and for $i=1,2$, $\langle e_{i}, e_s \rangle=3^{N-1}C_\text{min}$, the most negative inner product possible between two such effects. These inner products can only be achieved by neighboring adjacent and opposite adjacent effects respectively. $T^\dagger$ must preserve the inner product, therefore $T^\dagger (e_1)$ must be  neighboring and adjacent to $T^\dagger (e_2)$ on some subsystem $k'$, and $T^\dagger (e_3)$ must be adjacent and opposite to both $T^\dagger (e_1)$ and $T^\dagger (e_2)$, i.e.
\begin{align}
T^\dagger (e_1) &= \tilde{f}^{(1)} \otimes\cdots\otimes  \tilde{f}^{(k')}_i\otimes \cdots\otimes \tilde{f}^{(N)} \nonumber \\ 
T^\dagger (e_2) &= \tilde{f}^{(1)} \otimes\cdots\otimes  \tilde{f}^{(k')}_j\otimes \cdots\otimes \tilde{f}^{(N)} \nonumber \\ 
T^\dagger (e_3) &= \tilde{f}^{(1)} \otimes\cdots\otimes  \tilde{f}^{(k')}_l\otimes \cdots\otimes \tilde{f}^{(N)} ,
\end{align}
where $\tilde{f}_i \wedge \tilde{f}_j$, $\tilde{f}_l \vee \tilde{f}_{i}$ and $\tilde{f}_l \vee \tilde{f}_{j}$. As dim$(V^{(k)})=3$ and the set $\{\tilde{e}^{(k)}_1, \tilde{e}^{(k)}_2, \tilde{e}^{(k)}_s\}$ is linearly independent, all effects that are adjacent to $e_1$ on the $k^\text{th}$ subsystem lie in the linear span of $\{e_1, e_2, e_3\}$. By the linearity of $T^\dagger$, the image of any such effect is adjacent to $T^\dagger(e_1)$ on subsystem $k'$. Since $k$ was arbitrary, by Lemma \ref{adjacency-preserving-lemma} we have that $T^\dagger$ is trivial.
\end{proof}

It is worth noting how this proof breaks down when $n=3$ and the subsystems are classical trits. In this case, $C_\text{max} = C_\text{min} = 0$, so the inner product between all non-identical ray-extreme effects is zero and there can exist reversible transformations that do not preserve the adjacency of effects.

\end{document}